\documentclass{amsart}
\usepackage{amsmath}
\usepackage{amsfonts}
\usepackage{amssymb}
\usepackage{amsthm}
\usepackage{newlfont}
\usepackage{graphicx}
\usepackage{amscd}
\usepackage{bbm}

\theoremstyle{plain}
\newtheorem{Th}{Theorem}[section]
\newtheorem{Cor}[Th]{Corollary}

\newtheorem{Prop}[Th]{Proposition}
\theoremstyle{definition}

\theoremstyle{remark}
\newtheorem*{Rem}{Remark}%[section]
\numberwithin{equation}{section}
\newcommand{\VV}{{\mathbb V}}

\newcommand{\bOm}{\boldsymbol{\Omega}}

\newcommand{\bpsi}{\boldsymbol{\psi}}
\newcommand{\bphi}{\boldsymbol{\phi}}
\newcommand{\bsigma}{\boldsymbol{\sigma}}

\newcommand{\bom}{\boldsymbol{\omega}}

\begin{document}

\title[Discrete KP equation with self-consistent sources]{Discrete KP equation with self-consistent sources}

\author{Adam Doliwa}

\address{A. Doliwa: Faculty of Mathematics and Computer Science, University of Warmia and Mazury in Olsztyn,
ul.~S{\l}oneczna~54, 10-710~Olsztyn, Poland}
\email{doliwa@matman.uwm.edu.pl}
\urladdr{http://wmii.uwm.edu.pl/~doliwa/}

\author{Runliang Lin}

\address{R. Lin: Department of Mathematical Sciences, Tsinghua University, Beijing 100084, P.~R.~China}

\email{rlin@math.tsinghua.edu.cn}

%
%\date{\today}

\begin{abstract}
We show that the discrete Kadomtsev--Petviashvili (KP) equation with sources obtained recently by the ``source generalization" method can be incorporated into the squared eigenfunction symmetry extension procedure. Moreover, using the known correspondence between Darboux-type transformations and additional independent variables, we demonstrate that the equation with sources can be derived from Hirota's discrete KP equations but in a space of higher dimension. In this way we uncover the origin of the source terms as coming from multidimensional consistency of the Hirota system itself.
\end{abstract}
\maketitle

\section{Introduction}

\subsection{Background on soliton equation with self-consistent sources}

Soliton equations with self-consistent sources, 
proposed by Mel'nikov \cite{Melnikov1983},
have important applications in hydrodynamics, plasma physics and solid state physics 
(see, e.g., \cite{Doktorov-Shchesnovich-1995, Melnikov1983,Melnikov1984,Melnikov1987,Melnikov1989}).
It is known that the sources may change the velocities of the solitons \cite{Lin2001,Zeng2000},
similar observation can be found in \cite{Grinevich-Taimanov-2008,Melnikov1989capture}.
Later, some explicit solutions (such as, solitons, positons, negatons)
of some soliton equations with sources were obtained by using Darboux transformations 
(see the references in \cite{Lin2006}) and the Hirota method \cite{Zhang2003}.

Given an integrable system, its version with self-consistent sources is not unique. The most known such generalization of the
Kadomtsev--Petviashvili (KP) equation is \cite{Melnikov1983}
\begin{gather}
  \label{eqns:KPSCS}
    (4u_{,t}-12uu_{,x}-u_{,xxx})_{,x}-3u_{,yy}+4 ( \boldsymbol{r}^* \boldsymbol{q} )_{,xx}=0, \\
\label{eqns:KPSCS-qr}
    \boldsymbol{q}_{,y}  =\boldsymbol{q}_{,xx}+2u\boldsymbol{q} ,\qquad  
\boldsymbol{r}^*_{,y}  =-\boldsymbol{r}^*_{,xx}-2u\boldsymbol{r}^*, 
  \end{gather}
with the column-vector function $\boldsymbol{q} = (q_j)_{j=1, \dots , K}$ and the row-vector function $\boldsymbol{r}^* = (r^*_j)_{j=1, \dots ,K}$. It describes the interaction of a long wave with a short-wave packet propagating at an angle to each other. 
However, other more complicated extensions are also known~\cite{Hu2007,Melnikov1983}, and in order to find a unified framework to study such systems a systematical method was proposed  on the base of Sato's theory \cite{Liu2008a}; see also~\cite{MarekStefan} for a similar treatment of Gel'fand--Dikii hierarchies or \cite{Krichever-sources} for derivation of equations with sources as rational reductions of KP hierarchy.

Recall that the KP hierarchy~\cite{DKJM,MiwaJimboDate} reads
  \begin{equation} \label{eq:KP-h}
    L_{,t_n}=\bigl[ L^n_+,L \bigr], \qquad \text{where} \qquad 
L=\partial + \sum_{i=1}^\infty u_i\partial^{-i}
  \end{equation}
and ``$+$'' sign in subscript part of $L^n_+$ indicates the projection to the non-negative part of $L^n$ with respect to the powers of $\partial$. It is known that it allows for a squared eigenfunction symmetry (or ``ghost flow'') \cite{Oevel1993,Aratyn1998}
  \begin{equation}
    L_{,z} = \Bigl[\sum_{j=1}^K q_j\partial^{-1} r^*_j, L \Bigr],
\end{equation}
where $\boldsymbol{q}$ and $\boldsymbol{r}^*$ are solutions of the linear problem for the KP hierarchy and its adjoint
\begin{equation} \label{eq:KP-h-qr}
    \boldsymbol{q}_{,t_n} = L^n_+(\boldsymbol{q}), \qquad \boldsymbol{r}^*_{,t_n}  = - (L^n_+)^*(\boldsymbol{r}^* ),   
  \end{equation}
correspondingly. The idea to generate the KP hierarchy with self-consistent sources 
is to modify a specific flow (say $t_k$-flow, whose modification will be denoted by $\tilde{t}_k$) by the squared eigenfunction symmetry as
\begin{equation}
     L_{, \tilde{t}_k} = \Bigl[L^k_+ + \sum_{j=1}^K q_j\partial^{-1} r^*_j, L \Bigr],
\end{equation}
keeping in equations~\eqref{eq:KP-h-qr} all flows except of the $t_k$-flow.
In particular, equations \eqref{eqns:KPSCS}-\eqref{eqns:KPSCS-qr} follow from identification: $u=u_1$, $x=t_1$, $y=t_2$, $t=\tilde{t}_3$.

This systematic method has been used to generate extensions with self-consistent sources of the
CKP \cite{Wu2008}, multicomponent KP~\cite{HLYZ}, and some other hierarchies. A generalized dressing method has been also derived for these soliton hierarchies with sources, and some soliton solutions were obtained \cite{Liu2009}. Recently, a bilinear identity for the KP hierarchy with sources and
their Hirota bilinear equations were obtained \cite{Lin2013}.

\subsection{Discrete KP hierarchy and the Hirota equation}
By replacing in the Sato approach the partial differential $\partial$ by the partial difference operator $\Delta$ one arrives ~\cite{Kupershmidt-A,Kanaga-Tamizhmani} to a differential--difference KP hierarchy, which allows~\cite{Oevel} for the squared eigenfunction symmetry and gives in consequence self-consistent sources extensions \cite{YLZ}. Analogously one can obtain $q$-deformed KP hierarchy \cite{Lin2008,Lin2010} with sources. 

A fully discrete KP hierarchy was proposed in \cite{OHTI}. In \cite{KNW} the hierarchy was obtained from the Sato-like approach, and it was confirmed (as conjectured in~\cite{Zabrodin}) that all equations of the hierarchy can be obtained from Hirota's discrete KP equations \cite{Hirota}
\begin{equation} \label{eq:H-M}
\tau_{(i)}\tau_{(jk)} - \tau_{(j)}\tau_{(ik)} + \tau_{(k)}\tau_{(ij)} = 0,
\qquad 1\leq i< j <k ,
\end{equation}
here and in all the paper we use the short-hand notation with indices in brackets meaning shifts in discrete variables, for example $
\tau_{(\pm i)}(n_1, \dots , n_i, \dots  ) = \tau(n_1, \dots , n_i \pm 1, \dots )$.
The Hirota equations \eqref{eq:H-M} have a special position among discrete integrable systems (see for example reviews \cite{Zabrodin,KNS-rev}) both on the classical and the quantum level. In particular, as was shown by Miwa~\cite{Miwa} a single Hirota equation encodes the full KP hierarchy. 
A crucial property of the Hirota equation is that the number of independent variables can be arbitrary large, and such an extension does not create inconsistency or multivaluedness. This property, called multidimensional consistency \cite{ABS,FWN-cons}, is nowadays placed at the central point of discrete integrability theory and is considered as the precise analogue of exstence of a \emph{hierarchy} of nonlinear evolution equations in the case of continuous systems. 

Recently, a ``source generalization" method was proposed in \cite{Hu2006} which is based on replacing arbitrary constants in multisoliton solutions of an integrable equation without sources by arbitrary functions of one variable, and looking then for coupled bilinear equations whose solutions are those expressions. The method was applied there to Hirota's discrete KP equation producing the following system
\begin{equation} \label{eq:Hirota-sources-123}
\tau_{(1)}\tau_{(23)} - \tau_{(2)}\tau_{(13)} + \tau_{(3)}\tau_{(12)} = \boldsymbol{\rho}^*_{(13)} 
\bsigma_{(2)}, 
\end{equation}
where the column-vector function $\bsigma = (\sigma_j)_{j=1, \dots ,K}$, and the row-vector function 
$\boldsymbol{\rho}^* = (\rho_j^*)_{j=1, \dots , K}$ satisfy 
\begin{equation} \label{eq:Hirota-sources-13}
\tau_{(3)} \bsigma_{(1)} - \tau_{(1)} \bsigma_{(3)} = \bsigma \tau_{(13)} , \qquad 
\tau_{(1)} \boldsymbol{\rho}^*_{(3)} - \tau_{(3)} \boldsymbol{\rho}^*_{(1)} = \tau \boldsymbol{\rho}^*_{(13)}. 
\end{equation}

The original motivation of the paper was to reinterpret this result from the squared eigenfunction symmetry point of view, and this is the subject of Section~\ref{sec:sources-bin}, where we use relation~\cite{Oevel} between the discrete squared eigenfunction symmetry and vectorial binary Darboux transformation. Then in Section~\ref{sec:sources-lin} we present the linear problems for equations~\eqref{eq:Hirota-sources-123} with the corresponding
binary Darboux transformation interpretation.  Finally, in Section~\ref{sec:dim-sources}, using the known meaning of Darboux-type transformations as generators of additional independent discrete variables~\cite{LeBen}, we demonstrate that after an appropriate change of independent coordinates the equation with sources becomes just a system of Hirota's discrete KP equations. The number of additional dimensions depends on the number of source functions.

\section{The binary Darboux transformation flow as the source generation procedure}
\label{sec:sources-bin}
The Hirota system provides the compatibility condition for the linear problem~\cite{DJM-II}
\begin{equation} \label{eq:lin-dKP}
\bpsi_{(i)} - \bpsi_{(j)} = \frac{\tau \tau_{(ij)}}{\tau_{(i)} \tau_{(j)}} \bpsi ,  \qquad 1\leq i < j ,
\end{equation}
or its adjoint
\begin{equation} \label{eq:lin-dKP*}
\bpsi^*_{(j)} - \bpsi^*_{(i)} = \frac{\tau \tau_{(ij)}}{\tau_{(i)} \tau_{(j)}}  \bpsi^*_{(ij)} , 
\qquad1\leq i < j .
\end{equation}
There exists~\cite{Saito-Saitoh,Saitoh-Saito} an important duality between the linear problems and the Hirota equation itself.
\begin{Cor} \label{cor:psi-phi}
The functions 
\begin{equation} \label{eq:psi-phi}
\bphi = \tau \bpsi, \qquad \bphi^* = \tau \bpsi^*
\end{equation}
satisfy the following bilinear form of the linear problem and its adjoint
\begin{align} \label{eq:lin-bil}
\tau_{(j)} \bphi_{(i)} - \tau_{(i)} \bphi_{(j)} & = \bphi \tau_{(ij)} , \qquad i<j \\
\label{eq:ad-lin-bil}
\tau_{(i)} \bphi^*_{(j)} - \tau_{(j)} \bphi^*_{(i)} & = \tau \bphi^*_{(ij)}. 
\end{align}
\end{Cor}

Let us recall also~\cite{Nimmo-KP}, using notation of~\cite{Doliwa-Nieszporski}, the neccessary background on binary Darboux transformations of the Hirota system.
\begin{Th} \label{th:vect-Darb-KP}
Given the solution (column vector) $\boldsymbol{\omega}:\mathbb{Z}^{N}\to\VV$ 
of the linear system \eqref{eq:lin-dKP}, and given the solution (row vector)
$\boldsymbol{\omega}^*:\mathbb{Z}^{N}\to(\VV)^*$ of the adjoint linear system 
\eqref{eq:lin-dKP*},
construct the linear operator valued potential 
$\boldsymbol{\Omega}[\boldsymbol{\omega},\boldsymbol{\omega}^*]:
\mathbb{Z}^{N}\to \mathrm{L}(\VV)$,
defined by the system of compatible equations
\begin{equation} \label{eq:omega-p-p}
 \Delta_i \boldsymbol{\Omega}[\boldsymbol{\omega},\boldsymbol{\omega}^*] = 
\boldsymbol{\omega} \otimes \boldsymbol{\omega}^*_{(i)}, 
\qquad i = 1,\dots , N,
\end{equation} 
where $\Delta_i$ is the standard partial difference operator in direction of $n_i$.
Then (the binary Darboux transform of) the $\tau$-function 
\begin{equation}  
\tilde{\tau} = \tau \det \boldsymbol{\Omega}
[\boldsymbol{\omega},\boldsymbol{\omega}^*]  \label{eq:hat-tau-dKP}
\end{equation}
satisfies the Hirota equation~\eqref{eq:H-M} again.
\end{Th}
\begin{Rem}
The binary Darboux transformation provides a symmetry of the Hirota equation. Its infi\-ni\-te\-si\-mal version on the level of the KP hierarchy is provided by the squared eigenfunction symmetry.
\end{Rem}
\begin{Cor} \label{cor:det-Om}
We will need the following consequence of equations \eqref{eq:omega-p-p}:
\begin{equation} \label{eq:ev-det-Om}
\left( \det \boldsymbol{\Omega}
[\boldsymbol{\omega},\boldsymbol{\omega}^*] \right)_{(i)} = \det \boldsymbol{\Omega}
[\boldsymbol{\omega},\boldsymbol{\omega}^*] \left( 1 + \boldsymbol{\omega}^*_{(i)}  
\boldsymbol{\Omega}
[\boldsymbol{\omega},\boldsymbol{\omega}^*]^{-1} \boldsymbol{\omega} \right) .
\end{equation}
\end{Cor}
\begin{Cor} \label{cor:Om-psi-psi*}
Define the potentials $\bOm[\boldsymbol{\psi},\bom^*]$ and $\bOm[\bom, \boldsymbol{\psi}^*]$ by analogs of equations 
\eqref{eq:omega-p-p}, i.e.,
\begin{equation} \label{eq:Om-psi-psi*}
\Delta_i \bOm[\boldsymbol{\psi},\bom^*] = \boldsymbol{\psi} \otimes \bom^*_{(i)}, \qquad
\Delta_i \bOm[\bom, \boldsymbol{\psi}^*] =  \bom \otimes  \boldsymbol{\psi}^*_{(i)}.
\end{equation}
If the potential $\boldsymbol{\Omega}[\bom,\bom^*]$ is invertible, then
\begin{align} \label{eq:Psi-tr}
\tilde{\boldsymbol{\psi}} & = \boldsymbol{\psi} -
\bOm[\boldsymbol{\psi},\bom^*]  \bOm[\bom,\bom^*]^{-1}\bom ,\\
\label{eq:Psi*-tr}
\tilde{\boldsymbol{\psi}}^* & = \boldsymbol{\psi}^* - 
\bom^* \bOm[\bom,\bom^*]^{-1} \bOm[\bom, \boldsymbol{\psi}^*],
\end{align}
provide corresponding solutions of the transformed linear problem and its adjoint.
\end{Cor}
\begin{Rem}
Recall that in the proof \cite{Nimmo-KP} of Theorem~\ref{th:vect-Darb-KP} and Corollaries~\ref{cor:det-Om}-\ref{cor:Om-psi-psi*} one makes use of the so-called bordered determinant formula~\cite{Hirota-book} 
\begin{equation} \label{eq:bord-det}
\det \left( M^q_p - x^q y^*_p \right)=  \left| \begin{array}{cc} M^q_p & x^q \\ y^*_p & 1 \end{array} \right|
= \left| \begin{array}{cc} \boldsymbol{M} & \boldsymbol{x} \\ \boldsymbol{y}^* & 1 \end{array} \right|,
\end{equation}
where $\boldsymbol{M} = \left( M^q_p  \right)_{p,q=1,\dots ,K} $ is a square matrix, $\boldsymbol{x} = (x^q)_{q=1, \dots,K}$ is a column vector, and $\boldsymbol{y}^* = (y^*_p)_{p=1,\dots,K}$ is a row vector. Then equations \eqref{eq:Psi-tr} and \eqref{eq:Psi*-tr} for scalar functions $\tilde{\boldsymbol{\psi}}$ and $\tilde{\boldsymbol{\psi}}^*$ can be written in the form
\begin{equation*}
\tilde{\boldsymbol{\psi}} = \left| \begin{array}{cc} \bOm[\bom,\bom^*] & \bom \\
\bOm[\boldsymbol{\psi},\bom^*] & \boldsymbol{\psi} \end{array}\right| \cdot \Bigl|  \bOm[\bom,\bom^*] \Bigr|^{-1}, \qquad
\tilde{\boldsymbol{\psi}}^* = \left| \begin{array}{cc} \bOm[\bom,\bom^*] & \bOm[\bom,\boldsymbol{\psi}^*] \\
 \bom^* & \boldsymbol{\psi}^* \end{array}\right| \cdot \Bigl|  \bOm[\bom,\bom^*] \Bigr|^{-1},
\end{equation*} 
used in other parts of the paper.

\end{Rem}
\begin{Cor} \label{cor:t-om}
The function
\begin{equation} \label{eq:bD-om}
\tilde{\boldsymbol{\pi}} = \bOm[\bom,\bom^*]^{-1} \bom
\end{equation}
satisfies also the transformed linear problem. By $\boldsymbol{\pi}$ denote the corresponding solution of the initial linear problem~\eqref{eq:lin-dKP}.
\end{Cor}
\begin{Rem}
Formula \eqref{eq:bD-om} can be formally obtained from a transformation of the trivial (vector valued) solution $\boldsymbol{\psi} = \boldsymbol{0} \in\VV$ of the linear problem with the constant square matrix $\bOm[\boldsymbol{0},\bom^*] = - \mathbbm{1}$ in \eqref{eq:Psi-tr}. See however Proposition~\ref{prop:pi} of Section~\ref{sec:dim-sources} for another interpretation.
\end{Rem}

Notice that one can formally rewrite transformation formula \eqref{eq:Psi-tr} as
\begin{equation}
\tilde{\boldsymbol{\psi}}  = \boldsymbol{\psi} -
\Delta_i^{-1} (\boldsymbol{\psi} \otimes \bom^*_{(i)} ) \; \tilde{\boldsymbol{\pi}} ,
\end{equation}
which gives a discrete analogue of the squared eigenfunction symmetry (see also~\cite{Zabrodin}, where in the scalar $\dim\mathbb{V}=1$ case it is called the discrete adjoint flow). The general idea of getting equations with self-consistent sources as modification of a usual flow by the squared eigenfunction flow motivates to formulate the following result.
\begin{Th} \label{th:source-binary}
Pick up one direction, $n_i$ say, and modify the $i$th discrete flow by the above binary Darboux tansformation as follows:
\begin{equation} \label{eq:HM-sources}
\tau_{(\tilde{i})} = \tilde{\tau}_{(i)}, \qquad \text{etc.}
\end{equation}
Then in the new discrete coordinates the Hirota system is replaced by the system with self-consistent sources
\begin{equation} \label{eq:Hs-tau}
\tau_{(\tilde{i})} \tau_{(jk)} - \tau_{(j)} \tau_{(\tilde{i}k)} + \tau_{(k)} \tau_{(\tilde{i}j)}  = 
- (\tau \bom^*)_{(jk)} ( \tau \boldsymbol{\pi})_{(\tilde{i})} ,
\qquad i< j < k,
\end{equation}
with 
\begin{align}
\label{eq:Hs-om}
\boldsymbol{\pi}_{(j)} - \boldsymbol{\pi}_{(k)} & = \boldsymbol{\pi} \frac{\tau \tau_{(jk)}}{\tau_{(j)} \tau_{(k)}},\\
\label{eq:Hs-om*}
\bom^*_{(k)} - \bom^*_{(j)} & = \bom^*_{(jk)} \frac{\tau \tau_{(jk)}}{\tau_{(j)} \tau_{(k)}},
\end{align}
while $\boldsymbol{\pi}$ and the transformation datum $\bom$ are related by equation \eqref{eq:bD-om}.
\end{Th}
\begin{proof}
The left hand side of equations~\eqref{eq:Hs-tau} after using \eqref{eq:hat-tau-dKP} reads
\begin{equation*}
LHS = \tau_{(i)}\tau_{(jk)} (\det \boldsymbol{\Omega})_{(i)} - \tau_{(j)}\tau_{(ik)} (\det \boldsymbol{\Omega})_{(ik)}  + 
\tau_{(k)}\tau_{(ij)} (\det \boldsymbol{\Omega})_{(ij)} ,
\end{equation*}
where we abbreviate $\boldsymbol{\Omega} = \bOm[\bom,\bom^*]$. After making use of the Hirota system \eqref{eq:H-M}, the evolution rule \eqref{eq:omega-p-p} of the potential $\boldsymbol{\Omega}$ and the corresponding evolution rule 
\eqref{eq:ev-det-Om} of its determinant we obtain
\begin{equation*}
LHS = \frac{\tau_{(i)} \tau_{(j)} \tau_{(k)} }{\tau}  \left( 
\frac{ - \tau \tau_{(ik)} }{ \tau_{(i)} \tau_{(k)} } \bom^{*}_{(ik)}  + 
\frac{ \tau \tau_{(ij)} }{ \tau_{(i)} \tau_{(j)} } \bom^{*}_{(ij)} 
\right) \left( (\det \boldsymbol{\Omega}) \; \boldsymbol{\Omega}^{-1} \; \bom \right)_{(i)}.
\end{equation*}
Finally, the adjoint linear problem \eqref{eq:lin-dKP*} satisfied by the transformation datum $\bom^*$, and the transformation rule \eqref{eq:hat-tau-dKP} of the $\tau$-function give the right hand side of equations~\eqref{eq:Hs-tau}. 
\end{proof}
\begin{Rem}
In the formulation of Theorem~\ref{th:source-binary} we assumed that the index $i$ of the modified direction satisfies $i<j<k$. When we modify other index we correspondingly modify also the sign of the right hand side, i.e.
\begin{align} \label{eq:Hs-comp}
\tau_{(i)}\tau_{(\tilde{j}k)} - \tau_{(\tilde{j})}\tau_{(ik)} + \tau_{(k)}\tau_{(i\tilde{j})} & = \; \; 
(\tau \bom^*)_{(ik)}  ( \tau \boldsymbol{\pi})_{(\tilde{j})} ,\\
\tau_{(i)}\tau_{(j\tilde{k})} - \tau_{(j)}\tau_{(i\tilde{k})} + \tau_{(\tilde{k})}\tau_{(ij)} & = 
- (\tau \bom^*)_{(ij)}  ( \tau \boldsymbol{\pi})_{(\tilde{k})} .
\end{align}
Define functions 
\begin{equation} \label{eq:sigma-omega}
\bsigma = \tau \boldsymbol{\pi}, \qquad  \boldsymbol{\rho}^* = \tau \bom^*,
\end{equation}
then taking in equation~\eqref{eq:Hs-comp} $i=1$, $\tilde{j} = 2$ and $k=3$, by Corollary~\ref{cor:psi-phi} we obtain equations 
\eqref{eq:Hirota-sources-123}-\eqref{eq:Hirota-sources-13} of \cite{Hu2006}.  
\end{Rem}
\begin{Rem}
Needless to say, in any triple of non-modified variables we have the Hirota system \eqref{eq:H-M}.
\end{Rem}

\section{The linear problem for the Hirota equation with self-consistent sources}
\label{sec:sources-lin}
In this Section we give the linear problem (and the adjoint linear problem) for the main part of the discrete KP equation with sources. We keep the notation $\tilde{i}$ for the discrete modified flow, but we present its interpretation within the squared eigenfunction (binary Darboux transformation) symmetry only after we checked the postulated form of the linear problem.
\begin{Prop}
Equation \eqref{eq:Hs-tau} of the discrete KP system with sources
is the compatibility condition of the following linear system (assuming $i<j<k$):
\begin{align} \label{eq:lin-s-1}
\bpsi_{(\tilde{i})} - \bpsi_{(j)} & = 
\bpsi \frac{\tau \tau_{(\tilde{i}j)}}{\tau_{(\tilde{i})} \tau_{(j)}} - \bOm[\bpsi,\bom^*]_{(j)} \boldsymbol{\pi}_{(\tilde{i})},  \\
\label{eq:lin-s-2}
\bpsi_{(j)} - \bpsi_{(k)} & = \bpsi \frac{\tau \tau_{(jk)}}{\tau_{(j)} \tau_{(k)}},
\end{align}
with $\bOm[\bpsi,\bom^*]$ given by the following system of compatible equations:
\begin{equation}
\Delta_j \bOm[\bpsi,\bom^*]  =  \bpsi \otimes \bom^*_{(j)} \qquad j\neq i.
\end{equation}
\end{Prop}
\begin{proof}
Substract from equation \eqref{eq:lin-s-1} its version with index $j$ replaced by $k$, and add equation~\eqref{eq:lin-s-2}. Then the left hand side vanishes, and the right hand side, after making use of the definition of $\bOm[\bpsi,\bom^*]$ and equation~\eqref{eq:Hs-om*} gives equation \eqref{eq:Hs-tau}.
\end{proof}
\begin{Cor}
The adjoint linear problem reads
\begin{align} \label{eq:a-lin-s-1}
\bpsi^*_{(j)} - \bpsi^*_{(\tilde{i})} & = 
\bpsi^*_{(\tilde{i}j)} \frac{\tau \tau_{(\tilde{i}j)}}{\tau_{(\tilde{i})} \tau_{(j)}} + 
\bom^*_{(j)} \bOm[\boldsymbol{\pi},\bpsi^*]_{(\tilde{i})} ,\\
\label{eq:a-lin-s-2}
\bpsi^*_{(k)} - \bpsi^*_{(j)} & = \bpsi^*_{(jk)} \frac{\tau \tau_{(jk)}}{\tau_{(j)} \tau_{(k)}},
\end{align}
with $\bOm[\boldsymbol{\pi},\bpsi^*]$ given by the following system of compatible equations
\begin{equation}
\Delta_j \bOm[\boldsymbol{\pi},\bpsi^*] =  \boldsymbol{\pi} \otimes \bpsi^*_{(j)}, \qquad j\neq i.
\end{equation}
\end{Cor}
\begin{proof}
Substract from equation \eqref{eq:a-lin-s-1} shifted in $k$ its version with $j$ and $k$ exchanged, and add equation~\eqref{eq:a-lin-s-2} shifted in $\tilde{i}$. Then the left hand side vanishes, and the right hand side, after making use of the definition of $\bOm[\boldsymbol{\pi},\bpsi^*]$ and equation \eqref{eq:a-lin-s-2} once again gives equation \eqref{eq:Hs-tau}.
\end{proof}
\begin{Rem}
In checking the compatibility conditions above we didn't use equations \eqref{eq:Hs-om}. However, if we substract from equation \eqref{eq:lin-s-1} shifted in $k$ its version with $j$ and $k$ exchanged, and add equation~\eqref{eq:lin-s-2} shifted 
in $\tilde{i}$, then we obtain equation \eqref{eq:Hs-tau} but using this time also equations \eqref{eq:Hs-om}.  
\end{Rem}

Let us derive  \eqref{eq:lin-s-1}, the novel  part of the linear system for equations \eqref{eq:Hs-om} which is different from  \eqref{eq:lin-dKP}, starting from the binary Darboux transformation \eqref{eq:Psi-tr} of $\psi$. After some calculation using formula \eqref{eq:ev-det-Om} we obtain
\begin{equation}
\tilde{\bpsi}_{(i)} - \bpsi_{(j)} = \bpsi \frac{\tau \tau_{(ij)}}{\tau_{(i)} \tau_{(j)}}
\left( \frac{(\det \bOm[\bom,\bom^*])_{(j)}}{\det \bOm[\bom,\bom^*]}  \right)_{(i)} - 
\bOm [\psi,\bom^*]_{(j)} \tilde{\boldsymbol{\pi}}_{(i)},
\end{equation}
which after using transformation formula \eqref{eq:hat-tau-dKP} and identification~\eqref{eq:HM-sources} is identical 
with equation~\eqref{eq:lin-s-1}.

To derive equation \eqref{eq:a-lin-s-1} starting from the binary Darboux transformation \eqref{eq:Psi*-tr} of $\boldsymbol{\psi}^*$ is only slightly more involved. The transformation formula \eqref{eq:Psi*-tr} leads directly to
\begin{equation} \label{eq:lp-a-s-t}
\bpsi^*_{(j)} - \tilde{\bpsi}^*_{(i)} = \tilde{\bpsi}^*_{(ij)} \frac{\tau \tau_{(ij)}}{\tau_{(i)} \tau_{(j)}} +
\bom^*_{(i)} \bOm[\bom,\bom^*]^{-1}_{(i)} \bOm[\bom,\bpsi^*]^{-1}_{(i)}  + 
\bom^*_{(ij)} \bOm[\bom,\bom^*]^{-1}_{(ij)} \bOm[\bom,\bpsi^*]^{-1}_{(ij)}
\frac{\tau \tau_{(ij)}}{\tau_{(i)} \tau_{(j)}}  .
\end{equation}
One can notice that
\begin{equation*}
\Delta_j \left( \bOm[\bom,\bom^*]^{-1} \bOm[\bom,\bpsi^*] \right) =
\bOm[\bom,\bom^*]^{-1} \bom \otimes \left( \boldsymbol{\psi}^* - 
\bom^* \bOm[\bom,\bom^*]^{-1} \bOm[\bom, \boldsymbol{\psi}^*] \right)_{(j)} ,
\end{equation*}
which due to equation~\eqref{eq:bD-om} allows for the interpretation 
\begin{equation}
\bOm[\bom,\bom^*]^{-1} \bOm[\bom,\bpsi^*] = \bOm[\tilde{\boldsymbol{\pi}} ,\tilde{\bpsi}^*] =
\tilde{\bOm}[\boldsymbol{\pi},\bpsi^*] .
\end{equation}
Then equation~\eqref{eq:lp-a-s-t} can be transformed, using equation \eqref{eq:ev-det-Om} and properties of $\bom^*$, into
\begin{equation}
\bpsi^*_{(j)} - \tilde{\bpsi}^*_{(i)} = \tilde{\bpsi}^*_{(ij)} \frac{\tau \tau_{(ij)}}{\tau_{(i)} \tau_{(j)}} 
\left( \frac{(\det \bOm[\bom,\bom^*])_{(j)}}{\det \bOm[\bom,\bom^*]}  \right)_{(i)} +
\bom^*_{(j)} \tilde{\bOm}[\boldsymbol{\pi},\bpsi^*]_{(i)} ,
\end{equation}
which after using formula \eqref{eq:hat-tau-dKP} and identification~\eqref{eq:HM-sources} is identical 
with equation~\eqref{eq:a-lin-s-1}.

\section{The discrete KP equation with sources from the standard Hirota system}
\label{sec:dim-sources}
It is known that in the system of Hirota equations~\eqref{eq:H-M} the number of independent variables can be arbitrarily large. Moreover, as in other fundamental discrete integrable systems, there is essentially no difference between a Darboux-type transformation and a step into an additional dimension in the parameter space. Such an observation for certain integrable systems like the discrete Darboux equations~\cite{TQL}, and the lattice potential modified Korteweg--de~Vries equation~\cite{FWN-Capel-KdV}, is one of roots of the present-day approach to integrability of discrete systems as the so-called multidimensional consistency~\cite{ABS,FWN-cons}. In the theory, the fundamental possibility of extending the number of independent variables of a given nonlinear system by adding its copies in different directions restates the Bianchi superposition principle for Darboux transformations.

Having shown that the discrete KP equation with sources can be interpreted within the squared eigenfunction symmetry (binary Darboux transformation) approach, the idea of embedding the equation into the Hirota system in an appropriate large number of dimensions is very natural. However some details still have to be worked out.

Let us split independent variables in the system of Hirota equations~\eqref{eq:H-M} into two parts:
\begin{enumerate}
\item the evolution variables with indices $i,j,k$, from $1$ to $N$, and
\item the transformation variables of the first type with indices $a$ from $N+1$ to $N + K$, and of the second type with indices $b$ from $N+K+1$ to $N+2K$.
\end{enumerate}
The order of parameters introduced above is a consequence of signs in the Hirota system~\eqref{eq:H-M} which is going to produce correct signs in the linear problems and other ingredients of the binary Darboux transformation. It is known from both algebraic~\cite{Nimmo-KP} and geometric~\cite{Dol-WCR-Hirota} considerations that the Darboux transformation corresponds to a forward shift in a discrete variable, while the adjoint Darboux transformation (which may be considered as an inverse of the Darboux transformation) corresponds to a backward shift.
By the duality~\cite{Saito-Saitoh,Saitoh-Saito} between the linear problems and the Hirota equation itself, which we already described in Corollary~\ref{cor:psi-phi} one obtains the following result.
\begin{Prop}
Define a row-vector $\bom^*$ and a column-vector $\bom$ with $K$ components
\begin{equation} \label{eq:om-om*-def-tau}
\omega^*_p = \frac{\tau_{(-a_p)}}{\tau}, \qquad \omega^q = \frac{\tau_{(b_q)} }{\tau}, 
\qquad  \text{where} \quad 1\leq p,q \leq K, \quad p = a_p - N, \quad q = b_q - N - K. 
\end{equation}
Then
\begin{enumerate}
\item the Hirota system for two evolution variables and one transformation variable takes the form of the linear system~\eqref{eq:lin-dKP} satisfied by $\bom$, and its adjoint~\eqref{eq:lin-dKP*} satisfied by $\bom^*$; 
\item the Hirota system for one evolution variable and two transformation variables of different types compared with equation
\eqref{eq:omega-p-p} allows for identification of the transformation potential matrix elements as follows:
\begin{equation} \label{eq:Om-pq-def-tau}
\Omega^q_p = \Omega[\omega^q,\omega^*_p] = \frac{\tau_{(-a_p,b_q)}}{\tau} .
\end{equation}
\end{enumerate}
\end{Prop}
The Hirota system in transformation variables gives rise to the Bianchi permutability principle for various Darboux transformations. 
For our needs it is enough to study the additional variable interpretation of superpositions of binary Darboux transformations. Before doing that let us recall a four variable Pl\"{u}cker form~\cite{Shinzawa} of the Hirota equation 
\begin{equation} \label{eq:H-M-sub}
\tau_{(ij)}\tau_{(kl)} - \tau_{(ik)}\tau_{(jl)} + \tau_{(il)}\tau_{(jk)} = 0, \qquad 1\leq i < j < k < l,
\end{equation}
which can be obtained directly from three copies of equation~\eqref{eq:H-M} for triplets $(i,j,k)$, $(i,j,l)$ and $(i,k,l)$. 
\begin{Rem}
Equation \eqref{eq:H-M-sub} was an important step in deriving \cite{Dol-Tampa} the discrete Darboux equations~\cite{BoKo,MQL} from the Hirota system.
\end{Rem}  
The procedure to provide a dictionary between the binary Darboux transformation formulas \eqref{eq:hat-tau-dKP}, \eqref{eq:Psi-tr}-\eqref{eq:Psi*-tr} and the Hirota equation in the transformation variables will involve standard calculations using the bordered determinants technique, as in \eqref{eq:bord-det}. To match with the known equations let us combine the transformation variables into pairs $(a_p,b_p)$, $p=1,\dots ,K$, and then fix the order within each group
\begin{equation*}
a_K < a_{K-1} < \dots < a_1 , \qquad  b_1 < b_2 < \dots < b_K,
\end{equation*}
which we will follow in making use of equation \eqref{eq:H-M-sub}. To make the notation shorter, instead of $\tau_{(-a_p,b_p)}$ we write $\tau_{[p]}$ (this notation extends also to other functions of the discrete variables). Moreover, we write $\tau_{[\overline{K}]}$ instead of $\tau_{[1,2,\dots ,K]}$.

\begin{Prop}
The solution $\tau$ of the Hirota system \eqref{eq:H-M} translated in $K$th order "binary" shift can be expressed in terms of its mixed first order binary shifts $\tau_{(-a_p, b_q)}$ in a form \eqref{eq:hat-tau-dKP} of the $K$th order binary transformation 
\begin{equation} \label{eq:Om-b-shift}
\tau_{[\overline{K}]} = \tau \det (\Omega[\omega^q,\omega^*_p]), \qquad p,q = 1,\dots ,K,
\end{equation}
with the transformation potentials given by \eqref{eq:om-om*-def-tau} and \eqref{eq:Om-pq-def-tau}. 
\end{Prop}
\begin{proof}
The case $K=1$ is just the definition \eqref{eq:Om-pq-def-tau} of the matrix element $\Omega^1_1= \Omega[\omega^1,\omega^*_1]$. To validate the induction step we apply the column elimination technique starting from the right lower-corner element $\Omega^K_K$
\begin{equation*}
\det \bOm = \Omega^K_K \left| \Omega^q_p - \Omega^q_K (\Omega^K_K)^{-1} \Omega^K_p \right| = \frac{1}{\tau} \left( \tau \left| \Omega^q_p \right| \right)_{[K]}, \qquad p,q = 1, \dots , K-1,
\end{equation*}
where we used the four variable Hirota equation \eqref{eq:H-M-sub} for indices $i=a_{K}$, $j=a_p$, $k=b_q$, $l=b_{K}$ written in the form 
\begin{equation}
\Omega^q_{p[K]} = \Omega^q_p - \Omega^q_K (\Omega^K_K)^{-1} \Omega^K_p , \qquad p,q = 1, \dots , K-1.
\end{equation}
\end{proof}

To present the discrete KP equation with self-consistent sources in the additional variable approach
let us uncover the meaning of the solution $\boldsymbol{\pi}$ of the linear problem defined by equation~\eqref{eq:bD-om}. 
\begin{Prop} \label{prop:pi}
Let the $K$-vector $\bom$ and the $K\times K$ matrix $\bOm[\bom,\bom^*]$ be given as in \eqref{eq:om-om*-def-tau} and 
\eqref{eq:Om-pq-def-tau}. Then the components $\pi^{r}$, $r=1,\dots ,K$ of solution $\boldsymbol{\pi}$ of the equation 
\begin{equation}
\boldsymbol{\pi}_{[\overline{K}]} = \bOm[\bom,\bom^*] \bom
\end{equation}
read
\begin{equation} \label{eq:pi-r}
\pi^r = (-1)^{K-r} \frac{\tau_{(a_r)}}{\tau}, \qquad r=1,\dots , K.
\end{equation} 
\end{Prop}  
\begin{proof}
By Cramer's formula
\begin{equation*}
\pi^r_{[\overline{K}]} = \frac{\det \bOm_r}{\det \bOm},
\end{equation*}
where the matrix $\bOm_r$ is obtained from $\bOm=\bOm[\bom,\bom^*]$ by replacing its $r$-th column by $\bom$. When $r>1$, to find 
$\det \bOm_r$  we extract the left upper-corner $\Omega^1_1$ and we eliminate other elements of the first row to obtain the determinant of the $(K-1)\times (K-1)$ matrix with the $(r-1)$th column consisting of
\begin{equation} \label{eq:om-q-1}
\omega^q_{[1]} = \omega^q - \Omega^q_1 (\Omega^1_1)^{-1} \omega^1, \qquad q=2,\dots ,K,
\end{equation}
and other elements of the form
\begin{equation} \label{eq:Om-q-1}
\Omega^q_{p[1]} = \Omega^q_p - \Omega^q_1 (\Omega^1_1)^{-1} \Omega^1_p, \qquad p, q=2,\dots ,K, \qquad p\neq r.
\end{equation}
Both equations \eqref{eq:om-q-1} and \eqref{eq:Om-q-1} are consequences of the three and four term 
 Hirota equations \eqref{eq:H-M} and \eqref{eq:H-M-sub}, respectively, written in appropriate variables.

Application of such left upper-corner row-reduction procedure $r-1$ times gives
\begin{equation*}
\det \bOm_r = \frac{ \tau_{[\overline{r-1}]}}{\tau}
\det \left( \begin{array}{cccc} \omega^r & \Omega^r_{r+1} & \cdots & \Omega^r_K \\
\vdots & \vdots & & \vdots \\
\omega^K & \Omega^K_{r+1} & \cdots & \Omega^K_K 
    \end{array} \right)_{[\overline{r-1}]}.
\end{equation*}
We then eliminate elements of the last column starting from the right lower-corner element $\Omega^K_K$ to get
\begin{equation*}
\det \left( \begin{array}{cccc} \omega^r & \Omega^r_{r+1} & \cdots & \Omega^r_K \\
\vdots & \vdots & & \vdots \\
\omega^K & \Omega^K_{r+1} & \cdots & \Omega^K_K 
    \end{array} \right) =  \Omega^K_K 
\det \left( \begin{array}{cccc} - \omega^r & \Omega^r_{r+1} & \cdots & \Omega^r_{K-1} \\
\vdots & \vdots & & \vdots \\
- \omega^{K-1} & \Omega^{K-1}_{r+1} & \cdots & \Omega^{K-1}_{K-1} 
    \end{array} \right)_{[K]} ,
\end{equation*}
where we used the follwing consequences of the Hirota equations \eqref{eq:H-M} and \eqref{eq:H-M-sub} 
\begin{align} \label{eq:om-q-K}
\omega^q_{[K]} & = - \omega^q + \Omega^q_K (\Omega^K_K)^{-1} \omega^K, \qquad q=r,\dots ,K-1, \\ 
\label{eq:om-q-p-K}
\Omega^q_{p[K]} & = \Omega^q_p - \Omega^q_K (\Omega^K_K)^{-1} \Omega^K_p, \qquad p = r+1 , \dots , K-1, \quad q=r,\dots ,K-1.
\end{align}
We apply $(K-r)$ times such left lower-corner column-reduction procedure, and we make use of equation~\eqref{eq:Om-b-shift} and of the natural identity
\begin{equation*}
\omega^r = \Omega^r_r \left(  \frac{\tau_{(a_r)}}{\tau} \right)_{[r]},
\end{equation*}
to conclude the proof.
\end{proof}

\begin{Rem}
The analogs of equations \eqref{eq:om-q-1} and \eqref{eq:om-q-K} for the solutions $\bom^*$ of the adjoint linear problem in their additional variables interpretation \eqref{eq:om-om*-def-tau} can be also obtained from the Hirota system \eqref{eq:H-M} and read
\begin{align*}
\omega^*_{p([1]} & = \; \;  \omega^*_p -  \omega^*_{1}(\Omega^{1}_{1})^{-1} \Omega^{1}_p, \qquad p=2, \dots , K,\\
\omega^*_{p([K]} & = -\omega^*_p +  \omega^*_{1} (\Omega^{K}_{K} )^{-1} \Omega^{K}_p, \qquad p=1, \dots , K-1.
\end{align*}
\end{Rem}
By putting the above ingredients together
we obtain the final result.
\begin{Prop}
Consider the Hirota system in $2K + N$ dimensions and define composite flow $\tilde{i}$ modifying the $i$th discrete flow, $1\leq i \leq N$, by the binary shift $[\overline{K}]$. Then, depending on the order of indices $i,j,k$ from $1$ to $N$ we have the discrete KP system with sources as in Theorem~\ref{th:source-binary} (or in the subsequent Remark), where $\bom^*$ and $\boldsymbol{\pi}$ are given in equations \eqref{eq:om-om*-def-tau} and \eqref{eq:pi-r}.
\end{Prop}

Finally, we provide the additional variables interpretation of the remaining part of the binary Darboux transformation formulas in the part related to transformations of the wave function 
$\boldsymbol{\psi}$ and the adjoint wave function $\boldsymbol{\psi}^*$. The proofs consist of direct verification.
\begin{Prop}
The functions 
\begin{align} \label{eq:Om-p-def}
\Omega_p = & \Omega [\boldsymbol{\psi}, \omega^*_p ]  = \boldsymbol{\psi}_{(-a_p)} \frac{ \tau_{(-a_p)}}{\tau}, \\
\label{eq:Om-q-def}
\Omega^q = & \Omega [\omega^q, \boldsymbol{\psi}^* ] = - \boldsymbol{\psi}^*_{(b_q)}  \frac{ \tau_{(b_q)}}{\tau} 
\end{align}
are indeed the potentials in the sense of equations \eqref{eq:Om-psi-psi*}, and the evolution of the wave function and its adjoint in the binary Darboux transformation direction reads
\begin{align} \label{eq:Om-q-bin}
\boldsymbol{\psi}_{(-a_p,b_q)} & =  \boldsymbol{\psi}  - \Omega[\boldsymbol{\psi}, \omega^*_p ]  (\Omega^q_p)^{-1}
\omega^q , \\ \label{eq:Om-p-bin}
\boldsymbol{\psi}^*_{(-a_p,b_q)} & =  \boldsymbol{\psi}^*  - \Omega[\omega^q, \boldsymbol{\psi}^* ] (\Omega^q_p)^{-1} 
\omega^*_p .
\end{align}
\end{Prop}
We have therefore constructed basic ingredients of the interpretation of binary transformation of the $\tau$-function and of the wave function $\boldsymbol{\psi}$ and its adjoint $\boldsymbol{\psi}^*$ which are:
\begin{enumerate}
\item elements $\omega^q$ and $\omega^*_p$ of the transformation data $\bom = (\omega^q)_{q=1,\dots ,K}$ and  
$\bom^* = (\omega^*_p)_{p=1,\dots ,K}$ in terms of the elementary shifts and the backward elementary shifts of the $\tau$-function \eqref{eq:om-om*-def-tau};
\item elements $\Omega^q_p$ of the matrix potential $\bOm[\bom,\bom^*] = (\Omega[\omega^q,\omega^*_p])_{p,q=1,\dots ,K}$ in terms of the mixed binary shifts of the $\tau$-function \eqref{eq:Om-pq-def-tau};
\item elements $\Omega_p$ of the matrix potential $\bOm[\boldsymbol{\psi},\bom^*] = (\Omega[\boldsymbol{\psi},\omega^*_p])_{p =1,\dots ,K}$ in terms \eqref{eq:Om-p-def} of the elementary backward shifts of the wave function $\boldsymbol{\psi}$ and of the $\tau$-function;
\item elements $\Omega^q$ of the matrix potential $\bOm[\bom,\boldsymbol{\psi}^*] = (\Omega[\omega^q, \boldsymbol{\psi}^*])_{q =1,\dots ,K}$ in terms \eqref{eq:Om-q-def} of the elementary forward shifts of the wave function $\boldsymbol{\psi}$ and of the $\tau$-function.
\end{enumerate} 

By reductions of the corresponding determinants and using formulas \eqref{eq:Om-q-bin}-\eqref{eq:Om-p-bin} we have the following result, which in conjunction with Section~\ref{sec:sources-lin} provides the additional variables interpretation of the linear problem for the discrete KP equation with sources.
\begin{Prop}
The $K$th order binary shifts of the wave function $\boldsymbol{\psi}$ and of the adjoint wave function $\boldsymbol{\psi}^*$
can be written as follows:
\begin{align} \label{eq:Psi-b-ev}
\boldsymbol{\psi}_{[\overline{K}]}  = \boldsymbol{\psi} -
\bOm[\boldsymbol{\psi},\bom^*]  \bOm[\bom,\bom^*]^{-1}\bom  = &\left| \begin{array}{cc} \bOm[\bom,\bom^*] & \bom \\
\bOm[\boldsymbol{\psi},\bom^*] & \boldsymbol{\psi} \end{array}\right| \cdot \Bigl|  \bOm[\bom,\bom^*] \Bigr|^{-1}, 
\\
\label{eq:Psi*-b-ev}
\boldsymbol{\psi}^*_{[\overline{K}]}  =  \boldsymbol{\psi}^* - 
\bom^* \bOm[\bom,\bom^*]^{-1} \bOm[\bom, \boldsymbol{\psi}^*] = & 
\left| \begin{array}{cc} \bOm[\bom,\bom^*] & \bOm[\bom, \boldsymbol{\psi}^*] \\
 \bom^* & \boldsymbol{\psi}^* \end{array}\right| \cdot \Bigl|  \bOm[\bom,\bom^*] \Bigr|^{-1},
\end{align}
which agrees with the transformation formulas \eqref{eq:Psi-tr}-\eqref{eq:Psi*-tr} provided we interpret the Darboux transformations as shifts in the additional transformation variables.
\end{Prop}
\begin{proof}
We will show only the $K$th order binary shift formula for the wave function $\boldsymbol{\psi}$ leaving to the reader analogous proof for $\boldsymbol{\psi}^*$. The case $K=1$ has been shown already in \eqref{eq:Om-q-bin}. To validate the induction step we first notice the following consequence of the linear problem \eqref{eq:lin-dKP}:
\begin{equation} \label{eq:bin-ev-O-p}
\Omega_{p[K]} = - \Omega_p + \Omega_K (\Omega^K_K)^{-1} \Omega^K_{p}, \qquad p=1,\dots ,K-1.
\end{equation}
Then we calculate the determinant in the numerator by the column reduction with respect to the element $\Omega^K_K$, that is
\begin{equation*}
\left| \begin{array}{ccc} \Omega^q_p & \Omega^q_{K} & \omega^q \\
\Omega^{K}_p & \Omega^{K}_{K} & \omega^{K} \\
\Omega_p & \Omega_{K} &  \boldsymbol{\psi} \end{array}\right| = 
\Omega^{K}_{K} \left| \begin{array}{cc} \Omega^q_p & \omega^q \\
\Omega_p & \boldsymbol{\psi} \end{array}\right|_{[K]} , \qquad p,q = 1,\dots ,K-1, 
\end{equation*}
using also equations
\eqref{eq:om-q-K}, \eqref{eq:om-q-p-K}, \eqref{eq:Psi-b-ev} and \eqref{eq:bin-ev-O-p}.
\end{proof}

\section{Conclusion and remarks}
In the paper we presented the discrete KP equation with self-consistent sources \cite{Hu2006} as coming from the (sourceless) discrete KP system of Hirota in multidimensions by suitable cut-off of a subspace in the full discrete variables space. Our approach provides also the corresponding linear problem and its adjoint. An important step in the derivation was the squared-eigenfunction approach to the equation with sources, which we also interpreted in the spirit of the binary Darboux transformations.  

Our result can be considered as a step towards the following research problems:
\begin{enumerate}
\item find the self-consistent source extensions by the squared-eigenfunction approach of other distinguished three dimensional integrable discrete equations and corresponding lattice maps (discrete Darboux equations and the quadrilateral lattice maps~\cite{BoKo,MQL}, discrete BKP equation~\cite{Miwa}, discrete CKP equation~\cite{Kashaev}, quadratic reductions of quadrilateral lattices~\cite{q-red}) using known binary Darboux transformations for these systems;
\item find the self-consistent source extensions of distinguished reductions of the Hirota system like the discrete (modified) Korteweg--de~Vries equation~\cite{FWN-Capel-KdV} or other discrete Gel'fand--Dikii type equations~\cite{FWN-GD};
\item using the known connection (on the continuous level) of the squared eigenfunction symmetry and the so called restricted flows of integrable hierarchies \cite{MarekStefan,Zeng-Li-93,Oevel1993a,Oevel1994} find discrete analogs of the restricted flow equations;
\item use known algebro-geometric or analytic techniques of getting solutions to the Hirota system and its distinguished reductions find corresponding solutions of their extensions with sources;
\item find non-commutative analogs of integrable discrete equations from their known sourceless versions~\cite{Nimmo-NCKP,Dol-Des,Dol-GD}.
\end{enumerate}

We would like to stress that the relation between discrete the KP equation with sources and the standard system of Hirota's discrete KP equations becomes elementary and visible on the level of discrete systems only. The present-day interest in integrable discrete systems is a reflection of the fact that in the course of a continuous limiting procedure which often ``brings artificial complications" \cite{Zabrodin}, various symmetries and relations between 
different discrete systems are lost or hidden. Our paper gives a new 
example supporting this claim, and shows once again the prominent role of 
Hirota's discrete KP equation in integrable
systems theory.

\section*{Acknowledgments}
The authors would like to thank an anonymous reviewer for useful editorial suggestions. The research of A.~D. was supported in part by the Polish Ministry of Science and Higher Education grant No.~N~N202~174739. 
R.~L. was supported in part by National Natural Science Foundation of China (11171175, 11201477).

\bibliographystyle{amsplain}

\end{document}